\documentclass[reqno]{amsart}
\usepackage{graphicx,times,amssymb,amsmath,bbm,cite}
\usepackage{algorithm}
\usepackage{enumerate}
\usepackage{latexsym}
\usepackage{multirow}
\evensidemargin=\oddsidemargin
\allowdisplaybreaks

\newtheorem{theorem}{Theorem}[section]

\newtheorem{definition}[theorem]{Definition}

\newcommand\remove[1]{}
\newcommand{\nc}{\newcommand}

\allowdisplaybreaks
\nc\bfa{{\boldsymbol a}}\nc\bfA{{\bf A}}\nc\cA{{\mathcal A}}
\nc\bfb{{\boldsymbol b}}\nc\bfB{{\boldsymbol B}}\nc\cB{{\mathcal B}}
\nc\bfc{{\boldsymbol c}}\nc\bfC{{\bf C}}\nc\cC{{\mathcal C}}
\nc\sC{{\mathscr C}}
\nc\bfd{{\boldsymbol d}}\nc\bfD{{\bfD}}
\nc\cD{{\mathcal D}}
\nc\bfe{{\boldsymbol e}}\nc\bfE{{\bf E}}\nc\cE{{\mathcal E}}
\nc\bff{{\boldsymbol f}}\nc\bfF{{\bf F}}\nc\cF{{\mathcal F}}
\nc\bfg{{\boldsymbol g}}\nc\bfG{{\bf G}}\nc\cG{{\mathcal G}}
\nc\bfh{{\boldsymbol h}}\nc\bfH{{\bf H}}\nc\cH{{\mathcal H}}
\nc\bfi{{\boldsymbol i}}\nc\bfI{{\bf I}}\nc\cI{{\mathcal I}}\nc\sI{{\mathscr I}}
\nc\bfj{{\boldsymbolj}}\nc\bfJ{{\bf J}}\nc\cJ{{\mathcal J}}
\nc\bfk{{\boldsymbolk}}\nc\bfK{{\bf K}}\nc\cK{{\mathcal K}}
\nc\bfl{{\boldsymboll}}\nc\bfL{{\bf L}}\nc\cL{{\mathcal L}}
\nc\bfm{{\boldsymbolm}}\nc\bfM{{\bf M}}\nc\cM{{\mathcal M}}
\nc\bfn{{\boldsymboln}}\nc\bfN{{\bf N}}\nc\cN{{\mathcal N}}
\nc\bfo{{\boldsymbolo}}\nc\bfO{{\bf O}}\nc\cO{{\mathcal O}}
\nc\bfp{{\boldsymbolp}}\nc\bfP{{\bf P}}\nc\cP{{\mathcal P}}
\nc\eP{{\EuScriptP}}\nc\fP{{\mathfrak P}}
\nc\bfq{{\boldsymbol q}}\nc\bfQ{{\bf Q}}\nc\cQ{{\mathcal Q}}
\nc\bfr{{\boldsymbol r}}\nc\bfR{{\bf R}}\nc\cR{{\mathcal R}}
\nc\bfs{{\boldsymbol s}}\nc\bfS{{\boldsymbol S}}\nc\cS{{\mathcal S}}
\nc\bft{{\boldsymbol t}}\nc\bfT{{\bf T}}\nc\cT{{\mathcal T}}
\nc\bfu{{\boldsymbol u}}\nc\bfU{{\bf U}}\nc\cU{{\mathcal U}}
\nc\bfv{{\boldsymbol v}}\nc\bfV{{\bf V}}\nc\cV{{\mathcal V}}
\nc\bfw{{\boldsymbol w}}\nc\bfW{{\bf W}}\nc\cW{{\mathcal W}}
\nc\bfx{{\boldsymbol x}}\nc\bfX{{\bf X}}\nc\cX{{\mathcal X}}
\nc\bfy{{\boldsymbol y}}\nc\bfY{{\bf Y}}\nc\cY{{\mathcal Y}}
\nc\bfz{{\boldsymbol z}}\nc\bfZ{{\bf Z}}\nc\cZ{{\mathcal Z}}

\begin{document}

\title{Polar codes for distributed hierarchical source coding}
\author[Min Ye]{Min Ye$^\ast$} 
\thanks{$^\ast$
Department of ECE and Institute for Systems Research, 
			University of Maryland, College Park, MD 20742, Email: yeemmi@gmail.com
Research supported in part by NSF grant CCF1217545.}
\author[A. Barg]{Alexander Barg$^{\ast\ast}$}\thanks{$^{\ast\ast}$
Department of ECE and Institute for Systems Research, University
of Maryland, College Park, MD 20742, and IITP, Russian Academy of
Sciences, Moscow, Russia. Email: abarg@umd.edu. Research supported
in part by NSF grants CCF1217545, CCF1217894, and NSA
98230-12-1-0260, Email: abarg@umd.edu.}

\begin{abstract}We show that polar codes can be used to achieve the rate-distortion functions in the problem 
of hierarchical source coding also known as the successive refinement problem. We also analyze the distributed 
version of this problem, constructing a polar coding scheme that achieves the rate distortion functions 
for successive refinement with side information.
\end{abstract}

\maketitle\section{Introduction: Hierarchical source coding}
Hierarchical source coding, also known as successive refinement of information, was introduced by Koshelev \cite{Koshelev80,Koshelev81} and Equitz and Cover \cite{Equitz91}. This problem is concerned with the construction of a source code for a discrete memoryless source $X$ with respect to a given distortion measure $d_1$ that can be further refined to represent the same source within another distortion measure $d_2$ so that both representations approach the best possible compression rates for the given distortion values. This property is also termed divisibility of sources and it has an obvious interpretation in the context of $\epsilon$-nets in metric spaces \cite{Koshelev94}. Koshelev found a sufficient condition for successive refinement in \cite{Koshelev80}, showing that the source is divisible if the coarse description $X_1$ is independent of $X$ given the fine description $X_2$,  and Equitz and Cover showed that this condition is also necessary. If the Markov condition is not satisfied, then attaining the rate-distortion functions for both descriptions is impossible, and the excess rate needed to represent the source was quantified by Rimoldi \cite{Rimoldi94}  (see also Koshelev  \cite{Koshelev80}). As observed in \cite{Equitz91}, successive refinement is a particular case of the multiple description problem for which the region of achievable rates was established 
by Ahlswede \cite{Ahlswede85} and El Gamal and Cover \cite{ElGamal91}. A constructive scheme based on codes with low-density generator matrices together
with message-passing encoding was presented by Zhang et al. \cite{Zhang09}.

A  version of the successive refinement problem that incorporates side information used to represent the
source  was considered by Steinberg and Merhav in \cite{Steinberg04}. In this problem, the side information is
expressed as a pair of random variables that form a Markov chain with the source random variables and are used
at the initial stage of representing the source and at the refinement stage, respectively. Paper \cite{Steinberg04}
found the minimum possible rates for reproducing the source at the given distortion levels in the presence of 
side information. 

The aim of this paper is to construct an explicit scheme for successive refinement for the aforementioned problems
using polar codes.
Polar codes were initially designed to support communication at rates approaching capacity of binary-input symmetric memoryless 
channels \cite{arikan2009}. Subsequently they were shown to approach optimal performance for a number of information-theoretic problems with two or more users. A sampling of results includes lossless and lossy source coding problems \cite{Arikan2010,kor09a}, multiple-access channels \cite{abb12}, the degraded wiretap channel \cite{Mahdavifar2011}, as well as a range of other problems that previously relied on random coding (see the recent preprint \cite{Mondelli14} for a more detailed overview of applications of polar codes). 
Recently Honda and Yamamoto \cite{Honda13} showed that it is possible to modify the construction of polar codes so that the coding scheme supports 
capacity-achieving communication for channels that are not necessarily symmetric.
This result paves way for new applications of polar codes such as achieving optimal rates for broadcast channels \cite{Mondelli14}. In this paper we note that the asymmetric polar coding scheme can be also used for achieving rate-distortion functions in a range of problems of multiterminal source beginning with the basic 
successive refinement problem and extending to its distributed version \cite{Steinberg04} as well as other related schemes.
The construction of polar codes for successive refinement is an easy combination of the ideas of \cite{Honda13} and the earlier construction of polar codes achieving the rate-distortion function \cite{kor09a}. 
The distributed successive refinement problem is somewhat more difficult because of the need to incorporate the side
information in the analysis of the decoder of polar codes at the representation stage. In this part we 
first construct a polar coding scheme for the general version of the Wyner-Ziv problem of distributed compression
with side information and then use it to address the case of successive refinement. Note that polar codes for
the particular case of the Wyner-Ziv problem in which the side information is additive were constructed in an earlier work \cite{kor09a}.

In Sect.~\ref{sect:prelim} we introduce notation for polar codes, while the remaining Sections \ref{sect:SR}, \ref{sect:WZ} are devoted to the
two versions of the successive refinement problem discussed above.

\section{Preliminaries on polar codes}\label{sect:prelim}
In this part we set up notations for our application of polar codes. 
Let $n=2^m$ for some $m\in \mathbb{N}$. We use the notation $[n]=\{1,2,\dots,n\}$ and use the shorthand
notation $X^n$ for the vector $(X_1,X_2,...,X_n).$ Similarly we write $X_i^j$ instead of
$(X_i,...,X_j)$ and use analogous notation for other vectors of random variables and their realizations.

Define the polarizing matrix (or the Ar{\i}kan transform matrix)
as $G_n=B_n F^{\otimes m}$, where $F=\text{\small{$\Big(\hspace*{-.05in}\begin{array}{c@{\hspace*{0.05in}}c}
    1&0\\1&1\end{array}\hspace*{-.05in}\Big)$}}$,
$\otimes$ is Kronecker product of matrices, and $B_n$ is a 	``bit reversal'' permutation 
matrix. In his landmark paper \cite{arikan2009}, Ar{\i}kan showed that given a binary-input
channel $W$, there is a sequence of linear codes, whose generator matrices are appropriately chosen from the rows of $G_n$, achieving the symmetric capacity of $W.$

Let $U$ be a random variable defined on $\{0,1\},$ let $V$ be a discrete random variable supported on a finite set $\cV,$ and let $P_{UV}$ be their joint distribution. Define the 
Bhattacharyya parameter $Z(U|V)$ as follows:
    $$
Z(U|V)=2\sum_{v \in \mathcal{V}} P_V(v) \sqrt{{P}_{U|V}(0|v) {P}_{U|V}(1|v)}.
    $$

Consider a binary random variable $X \sim {P}_X$ and let $X^n$ denote $n$ independent copies of $X.$ 
Consider random variables $U^n=(U_1,\dots,U_n)$ obtained from $X^n$ using the transformation $U^n=X^n G_n.$
Define the subsets $\mathcal{H}_X$ and $\mathcal{L}_X$ of $[n]$ as follows (definition of both sets depends on $n$, but for simplicity we omit $n$ in the notations):
\begin{equation}\label{def1}
    \begin{aligned}
\mathcal{H}_X&=\{i \in [n]:Z(U_i|U^{i-1}) \geq 1-\delta_n \} \\
\mathcal{L}_X&=\{i \in [n]:Z(U_i|U^{i-1}) \leq \delta_n \}
   \end{aligned}
\end{equation}
where $\delta_n\to 0$ as $n\to\infty$. In other words, $\cL_X$ consists of the indices for which the
bits $U_i$ are almost deterministic given the values of $U^{i-1},$ while $\cH_X$ includes bits that are
almost uniformly random given previous indices. The source polarization theorem of \cite{arikan2009,Arikan2010} asserts that
\begin{equation} \label{eq:cardi1}
 \begin{aligned}
&\lim_{n \to \infty} \frac{1}{n}|\mathcal{H}_{X}|=H(X) \\
&\lim_{n \to \infty} \frac{1}{n}|\mathcal{L}_{X}|=1-H(X)
\end{aligned}
\end{equation}
i.e., almost all indices fall into either $\mathcal{H}_X$ or $\mathcal{L}_X$. Moreover, the quantity
$\delta_n$ behaves as $2^{-n^{\beta}},$ where $\beta$ can be any constant that satisfies $0<\beta<1/2.$

Extension of these results to the case with side information can be phrased as follows. 
Let $(X,Y) \sim {P}_{XY}$ be a pair of finite discrete random variables, and assume that $X$ is binary. 
As before, let $U^n=X^n G_n.$ 
Define the index subsets $\mathcal{H}_{X|Y}$ and $\mathcal{L}_{X|Y}$ of $[n]$ as follows:
\begin{equation}\label{def2}
    \begin{aligned}
  \mathcal{H}_{X|Y}&=\{i \in [n]:Z(U_i|U^{i-1},Y^n) \geq 1-\delta_n \} \\
  \mathcal{L}_{X|Y}&=\{i \in [n]:Z(U_i|U^{i-1},Y^n) \leq \delta_n \}.
   \end{aligned}
\end{equation}
Similarly to \eqref{eq:cardi1} we have \cite{arikan2009}
\begin{equation} \label{eq:cardi2}
 \begin{aligned}
&\lim_{n \to \infty} \frac{1}{n}|\mathcal{H}_{X|Y}|=H(X|Y) \\
&\lim_{n \to \infty} \frac{1}{n}|\mathcal{L}_{X|Y}|=1-H(X|Y).
\end{aligned}
\end{equation}
By an application of the Cauchy-Schwarz inequality it is easy to see that
  $
  Z(U_i|U^{i-1},Y^n)\le Z(U_i|U^{i-1})
  $
and therefore the subsets defined above are related as follows:
\begin{equation} \label{relation}
 \begin{aligned}
    \mathcal{H}_{X|Y} &\subseteq \mathcal{H}_X \\
\mathcal{L}_X &\subseteq \mathcal{L}_{X|Y}.
\end{aligned}
\end{equation}
\remove{
\begin{equation*}
\begin{aligned}
Z(U_i|U^{i-1},Y^n) &=2 \sum_{u^{i-1},y^n} \sqrt{{P}_{U_i,U^{i-1},Y^n}(0,u^{i-1},y^n) {P}_{U_i,U^{i-1},Y^n}(1,u^{i-1},y^n)} \\
&=2 \sum_{u^{i-1}} \Big( \sum_{y^n}\sqrt{{P}_{U_i,U^{i-1},Y^n}(0,u^{i-1},y^n) {P}_{U_i,U^{i-1},Y^n}(1,u^{i-1},y^n)} \Big) \\
&\overset{\text{(a)}}{\leq} 2 \sum_{u^{i-1}} \Big( \sqrt{ (\sum_{y^n}{P}_{U_i,U^{i-1},Y^n}(0,u^{i-1},y^n))( \sum_{y^n}{P}_{U_i,U^{i-1},Y^n}(0,u^{i-1},y^n))} \Big) \\
&= 2 \sum_{u^{i-1}} \sqrt{{P}_{U_i,U^{i-1}}(0,u^{i-1}) {P}_{U_i,U^{i-1}}(1,u^{i-1})}
=Z(U_i|U^{i-1}),
\end{aligned}
\end{equation*}
where (a) follows from Cauchy–Schwarz inequality. Thus, for all $n$ 
\begin{equation*}
\begin{aligned}
\mathcal{H}_{X|Y} \subseteq \mathcal{H}_X, \\
\mathcal{L}_X \subseteq \mathcal{L}_{X|Y}.
\end{aligned}
\end{equation*}
}

\section{Successive refinement with polar codes}\label{sect:SR}

Let $X\sim P_X$ be a discrete memoryless source with a finite source alphabet $\cX.$ 
Let $\cT$ be a finite reproduction alphabet and let $d:\cX\times\cT\to[0,\infty)$ be a distortion function. 
 Let $R(D)$ be the rate distortion function given by $R(D)=\min_{P_{T|X}}I(X;T),$ where $P_{T|X}$ is such that
 $E_{XT}(d(X,T))\le D.$

Below we consider only the case of binary reproduction alphabets (extensions to other alphabets can be
easily accomplished based on the multiple methods available in the literature, e.g. \cite{Mori14,park13}). 
Suppose that there 
exist encoding functions
  \begin{align} \label{eq:sr1}
    &\phi_1:\cX^n\to [M_1]\\
   & \phi_2:\cX^n\to[M_2]\label{eq:sr1-1}
    \end{align}
    and decoding functions
 \begin{align}\label{eq:sr2}
   &\psi_1:[M_1]\to \cT^n\\
   &\psi_2:[M_1]\times[M_2]\to \cT^n\label{eq:sr2-1}   
   \end{align} 
such that 
  \begin{align}
   &E_{X^n}d(X^n,\psi_1(\phi_1(X^n)))\le D_1\label{eq:sr3}\\
   &E_{X^n}d(X^n,\psi_2(\phi_1(X^n),\phi_2(X^n)))\le D_2. \label{eq:sr31}
  \end{align}
  
Let  $M_1=2^{nR_1},M_2=2^{n(R_2-R_1)},$ where $(R_1,R_2)$ are the rate values for the two representations of the source $X.$
Given a distortion pair $(D_1,D_2),$ we say that the rate pair $(R_1,R_2)$ is achievable if for any $\epsilon_1>0,\epsilon_2>0,\delta>0$ there exists a sufficiently large $n=n(\epsilon_1,\epsilon_2,\delta)$ such that there exists a coding scheme satisfying \eqref{eq:sr1}-\eqref{eq:sr31} with block length $n$, rates not exceeding $R_1+\epsilon_1,R_2+\epsilon_1+\epsilon_2,$
and distortions $D_1+\delta,D_2+\delta.$

The source $X$ is said to be {\em successively refinable} with distortions $D_1$ and $D_2,$
$D_2\le D_1,$ if the pair of rate values $(R(D_1),R(D_2))$ is achievable.
The following result characterizes the set of achievable rate pairs.
\begin{theorem} \label{thm:sr} {\rm \cite{Koshelev80,Equitz91}}
Let $X$ be a source and let $T,W$ be two binary random variables.
The source is successively refinable if and only if there exists a conditional distribution 
${P}_{TW|X}$ with
$$
  E_{XT}d(X,T)\le D_1, \quad E_{XW}d(X,W)\le D_2
$$
\begin{equation} \label{eq:cond}
     I(X;T)=R(D_1),\quad I(X;W)=R(D_2)
\end{equation}
  and such that $X,W,T$ satisfy the Markov condition
\begin{equation} \label{eq:Markov}
X \rightarrow W \rightarrow T.
\end{equation}
\end{theorem}
The Markov property \eqref{eq:Markov} implies that $I(X;WT)=I(X;W)=R(D_2).$ 
Combined with \eqref{eq:cond}, we obtain
    \begin{align*}
H(W|T)-H(W|T,X)&=I(X;W|T)\\
&=I(X;W,T)-I(X;T)=R(D_2)-R(D_1).
     \end{align*}

\vspace*{.1in}Let ${P}_{TWX}={P}_{TW|X}{P}_X$ be the joint distribution of the triple $(T,W,X)$ that satisfies the conditions 
of the theorem. Let $(T^n,W^n,X^n)$ be a sequence of $n$ independent copies of the triple $(T,W,X)$. 
Define random vectors $U^n=T^nG_n$ and $V^n=W^nG_n.$ Below we use various conditional distributions derived from the joint
distribution ${P}_{U^nT^nV^nW^nX^n}.$ 

Define the index subsets $\mathcal{H}_T, \mathcal{L}_T, \mathcal{H}_{T|X}$, 
$\mathcal{L}_{T|X}, \mathcal{H}_{W|T}$, $\mathcal{L}_{W|T}$, $\mathcal{H}_{W|TX}$, $\mathcal{L}_{W|TX}$ in the way
analogous to \eqref{def1} and \eqref{def2}. For instance,
  \begin{gather*}
  \cH_{W|T}=\{i\in [n]: Z(V_i|V^{i-1},T^n)\ge 1-\delta_n\}\\
  \cH_{W|TX}=\{i\in[n]: Z(V_i|V^{i-1},X^n,T^n)\ge 1-\delta_n\}
  \end{gather*}
etc. Relationships analogous to \eqref{eq:cardi1}, \eqref{eq:cardi2}, \eqref{relation} 
hold. For instance, let $\mathcal{I}_T=(\mathcal{L}_T \cup \mathcal{H}_{T|X})^c$, $\mathcal{I}_W= (\mathcal{L}_{W|T} 
\cup \mathcal{H}_{W|TX})^c$, where $^c$ refers to the complement in $[n],$ then
    \begin{equation*}
\begin{aligned}
&\lim_{n \to \infty} \frac{1}{n}|\mathcal{I}_T|=I(X;T)=R(D_1) \\
&\lim_{n \to \infty} \frac{1}{n}|\mathcal{I}_W|=H(W|T)-H(W|T,X)=R(D_2)-R(D_1).
\end{aligned}
     \end{equation*}

Suppose that we are given the source sequence $x^n.$  
To construct the coding scheme, let us partition the set of indices according to
   \begin{equation}\label{eq:part}
     [n]=\cH_{T|X}\cup\cL_T\cup\cI_T.
   \end{equation}
Using the relations between the Bhattacharyya parameters and the corresponding entropies \cite{arikan2009},
we observe that in the successive cancellation coding scheme we should use the indices in the set $\cI_T.$ Indeed, if $i \in \cH_{T|X}$ then $Z(U_i|U_1^{i-1},X^n) \approx 1$ and
so $H(U_i|U_1^{i-1},X^n)\approx 1$ and
  $$
  I(U_i;X^n|U_1^{i-1})=H(U_i|U_1^{i-1})-H(U_i|U_1^{i-1},X^n) \approx0.
  $$
Therefore, the bits $u_i,i\in\cH_{T|X}$ are nearly independent of the source sequence conditional
on the previously found values $u_1^{i-1}.$ Likewise we observe that $I(U_i;X^n|U_1^{i-1}) \approx 0, i \in \cL_T,$ 
and so the bits indexed by $\cL_T$ are almost deterministic.
At the same time, if $i \in \cI_T$ then $Z(U_i|U_1^{i-1},X^n) \approx 0$ and $Z(U_i|U_1^{i-1}) \approx1,$
so 
   $$
   I(U_i;X^n|U_1^{i-1})=H(U_i|U_1^{i-1})-H(U_i|U_1^{i-1},X^n)\approx 1.
   $$

These considerations motivate the following encoding procedure.
First, if $i\in \cH_{T|X}$ then we put $u_i=0$ or $1$ with probability $1/2$ independently of the source sequence and each other. 
Following the accepted usage in polar codes, we call the values $u_i, i\in \cH_{T|X}$ frozen bits and assume that they are available both to the encoder and decoder.
Next assign the values $u_i,i\in \cI_T\cup\cL_T$ successively as follows. Assume that the sequence $u^{i-1}, i\ge 0$ has been chosen. If $i\in \cI_T$, choose $u_i$ in a randomized way according to the distribution
  \begin{equation}\label{eq:random1}
     \Pr(u_i=a)={P}_{U_i|U^{i-1},X^n}(a|u^{i-1},x^n), \quad a=0,1
 \end{equation}
and if $i\in \cL_T,$ put
  \begin{equation}\label{eq:func1}
    u_i=u_i(u^{i-1})\triangleq\arg\max_{a\in\{0,1\}} {P}_{U_i|U^{i-1}}(a|u^{i-1}) .
  \end{equation}
This concludes the description of the encoding function $\phi_1$ in \eqref{eq:sr1}.
  
The encoder $\phi_2$ relies on the sequence $u^n$ as well as the source sequence $x^n$ \eqref{eq:sr1-1} and is designed
as follows. We begin with finding $t^n=u^nG_n$ (note that $G_n^{-1}=G_n$) which is then used to compute the sequence $v^n$.
Partition the  set of coordinates as follows:
   $$
   [n]=\cH_{W|TX}\cup \cL_{W|T}\cup I_W.
   $$
The above arguments apply here as well. The bits $v_i,i\in \cH_{W|TX}$ are set to $0$ or $1$ with probability $1/2$ independently of $x^n,t^n$ and each other, and are made available both to the encoder and the decoder.
The values $v_i, i\in \cI_W$ are assigned randomly according to the distribution
   \begin{equation}\label{eq:random2}
   \Pr(v_i=a)={P}_{V_i|V^{i-1}T^nX^n}(a|v^{i-1},t^n,x^n), \quad a=0,1.
   \end{equation}
The values $v_i, i\in \mathcal{L}_{W|T}$   are assigned as follows:
\begin{equation}\label{eq:func2}
v_i=\arg\max_{a\in\{0,1\}} {P}_{V_i|V^{i-1}T^n}(a|v^{i-1},t^n).
\end{equation}
This concludes the description of the encoder $\phi_2.$

The information bits $u_{\cI_T}$ and $v_{\cI_W}$ are transmitted to the decoder.
In addition, the decoder knows the values of all the frozen bits. The first-layer mapping $\psi_1$ consists of finding the
values $u_{\mathcal{L}_T}$ using rule \eqref{eq:func1}. Upon completing this, the decoder knows all the bits $u^n$ and obtains 
the first-layer reproduction sequence  $t^n=u^n G_n.$ 
To construct the more refined representation of the source sequence $x^n$, the decoder uses $t^n$ to determine $v_{\mathcal{L}_{W|T}}$ from \eqref{eq:func2}. Upon completing this, the decoder can find the second-layer reproduction sequence $w^n=v^nG_n$. 

By construction we clearly obtain the desired values of the rates of the two-layer source codes: 
   $$
   R_1= \frac{|\mathcal{I}_T|} {n} \to R(D_1), \quad R_2= \frac{|\mathcal{I}_W|} {n} \to R(D_2)-R(D_1).
   $$
Turning to the distortion, suppose that the values of the frozen bits $u_{\mathcal{H}_{T|X}}$ and $v_{\mathcal{H}_{W|TX}}$ are fixed.
Then the average values of the distortion for the two representations of the source are given by
     \begin{equation} \label{dist1}
D_{1,n} (u_{\mathcal{H}_{T|X}})={E}_{X^n} \Big[ {E} [d(X^n,w^n(u_{\mathcal{H}_{T|X}},X^n)] \Big]
\end{equation}
\begin{equation} \label{dist2}
{D_{2,n}} (u_{\mathcal{H}_{T|X}},v_{\mathcal{H}_{W|TX}})= {E}_{X^n} \Big[ {E} [d(X^n,v^n(u_{\mathcal{H}_{T|X}},v_{\mathcal{H}_{W|TX}},X^n)] \Big]
\end{equation}
respectively, where the inner expectations in \eqref{dist1} and \eqref{dist2} are taken over randomization in \eqref{eq:random1} and \eqref{eq:random2}.
To show that there exists a choice of the frozen bits for which the values $D_{1n}$ and $D_{2n}$ approach $D_1$ and $D_2,$
we compute the average distortions over all the possible assignments of frozen bits.
\begin{theorem} Let $0 < {\beta}' < \beta < 1/2$. Then
  \begin{equation}\label{eq:ED}
 {E}[D_{1,n} (U_{\mathcal{H}_{T|X}})] \leq D_1+O(2^{-n^{{\beta}'}}), \quad {E}[D_{2,n} (U_{\mathcal{H}_{T|X}},V_{\mathcal{H}_{W|TX}})] \leq D_2+O(2^{-n^{{\beta}'}}).
 \end{equation}
  Consequently, there exists a choice of the frozen bits $u_{\mathcal{H}_{T|X}},v_{\mathcal{H}_{W|TX}}$ such that 
  $D_{1,n} (u_{\mathcal{H}_{T|X}})=D_1+O(2^{-n^{{\beta}'}})$ and $D_{2,n} (u_{\mathcal{H}_{T|X}},v_{\mathcal{H}_{W|TX}})=D_2+O(2^{-n^{{\beta}'}})$.
\end{theorem}
\begin{proof}
In the proofs below, we often omit the subscript random variables in the notations of distributions if the realizations are denoted by lowercase letters that identify them without ambiguity. For example, $P(u_i|u^{i-1},x^n)$ stands for $P_{U_i|U^{i-1}X^n}(u_i|u^{i-1},x^n),$ etc. Denote by ${Q}_{T^nW^nX^n}$ the joint distribution of the reproduction sequences and source sequence and let 
${Q}_{U^nX^n}$ and ${Q}_{V^nT^nX^n}$ be the distributions derived from it. We have
    \begin{align*}
{Q}_{U^nX^n}(u^n,x^n)={P}_{X^n}(x^n) &\Big( \prod_{i \in \mathcal{I}_T} {P}_{U_i|U^{i-1}X^n}(u_i|u^{i-1},x^n) \Big) 
2^{-|\mathcal{H}_{T|X}|}\\ 
&\times \Big( \prod_{i \in \mathcal{L}_T} \mathbbm{1} [{P}_{U_i|U^{i-1}}(u_i|u^{i-1}) > {P}_{U_i|U^{i-1}}(u_i \oplus 1|u^{i-1})] \Big)\\
{Q}_{V^nT^nX^n}(v^n,t^n,x^n)={Q}_{T^nX^n}(t^n,x^n) &\Big( \prod_{i \in \mathcal{I}_W} {P}_{V_i|V^{i-1}T^nX^n}(v_i|v^{i-1},t^n,x^n) 
\Big) 2^{-|\mathcal{H}_{W|TX}|}\\
&\times    \Big( \prod_{i \in \mathcal{L}_{W|T}} \mathbbm{1} [{P}_{V_i|V^{i-1}T^n}(v_i|v^{i-1},t^n) > {P}_{V_i|V^{i-1}T^n}(v_i \oplus 1|v^{i-1},t^n)] \Big).
\end{align*}
The $L_1$ distance between ${P}_{U^nX^n}$ and ${Q}_{U^nX^n}$ can be bounded as follows:
\begin{align*}
&\|{P}_{U^nX^n}-{Q}_{U^nX^n}\|_1 \\
&=\sum_{u^n,x^n}|{P}(u^n,x^n)-{Q}(u^n,x^n)|=\sum_{u^n,x^n}|({P}(u^n|x^n)-{Q}(u^n|x^n)){P}(x^n)| \\
&\overset{(a)}{=}\sum_{u^n,x^n} \Big| \sum_i ({Q}(u_i|u^{i-1},x^n)-{P}(u_i|u^{i-1},x^n)){P}(x^n) \Big( \prod_{j=1}^{i-1} {P}(u_j|u^{j-1},x^n) \Big) \Big( \prod_{j=i+1}^n {Q}(u_j|u^{j-1},x^n) \Big) \Big| \\
& \overset{(b)}{\leq}\sum_{i \in \mathcal{I}_T^c} \sum_{u^n,x^n} |{Q}(u_i|u^{i-1},x^n)-{P}(u_i|u^{i-1},x^n)|{P}(x^n) {P}(u^{i-1}|x^n) {Q}(u_{i+1}^n|u^i,x^n) 
\end{align*}
\begin{equation} \label{proof1}
\begin{aligned}
&=\sum_{i \in \mathcal{L}_T} \sum_{u^i,x^n} \Big| \mathbbm{1} [{P}(u_i|u^{i-1}) > {P}(u_i \oplus 1|u^{i-1})]-{P}(u_i|u^{i-1},x^n) \Big| {P}(u^{i-1},x^n)\\
 &+\sum_{i \in \mathcal{H}_{T|X}} \sum_{u^i,x^n} \Big| \frac{1}{2}-{P}(u_i|u^{i-1},x^n) \Big| {P}(u^{i-1},x^n) \\
&\overset{(c)}{=}\sum_{i \in \mathcal{L}_T} \sum_{u^{i-1},x^n} 2{P}(u_i(u^{i-1}) \oplus 1|u^{i-1},x^n) {P}(u^{i-1},x^n) + \sum_{i \in \mathcal{H}_{T|X}} 2{E} \Big| \frac{1}{2}-{P}(0|U^{i-1},X^n) \Big| \\
&\overset{(d)}{\leq} \sum_{i \in \mathcal{L}_T} \sum_{u^{i-1}} 2{P}(u_i(u^{i-1}) \oplus 1,u^{i-1}) +\sum_{i \in \mathcal{H}_{T|X}}2\sqrt{{E} \Big[ ( \frac{1}{2}-{P}(0|U^{i-1},X^n) )^2 \Big]}\\
&\overset{(e)}{\leq} \sum_{i \in \mathcal{L}_T} \sum_{u^{i-1}} 2 \sqrt{{P}(0,u^{i-1}){P}(1,u^{i-1})}+\sum_{i \in \mathcal{H}_{T|X}}2\sqrt{{E}\Big[\frac{1}{4}-{P}(0|U^{i-1},X^n){P}(1|U^{i-1},X^n)\Big]} \\
&=\sum_{i \in \mathcal{L}_T}Z(U_i|U^{i-1}) \\
&+\sum_{i \in \mathcal{H}_{T|X}}2\sqrt{{E}
\Big[\Big(\frac{1}{2}-\sqrt{{P}(0|U^{i-1},X^n){P}(1|U^{i-1},X^n)}\Big) \Big(\frac{1}{2}+\sqrt{{P}(0|U^{i-1},X^n){P}(1|U^{i-1},X^n)}\Big)\Big]}\\
&\overset{(f)}{\leq}\sum_{i \in \mathcal{L}_T}Z(U_i|U^{i-1})+\sum_{i \in \mathcal{H}_{T|X}}2\sqrt{{E}\Big[\frac{1}{2}-\sqrt{{P}(0|U^{i-1},X^n){P}(1|U^{i-1},X^n)}\Big]}\\
&=\sum_{i \in \mathcal{L}_T}Z(U_i|U^{i-1})+\sum_{i \in \mathcal{H}_{T|X}}2\sqrt{\frac{1}{2}-\frac{1}{2}Z(U_i|U^{i-1},X^n)}\\
&=O(2^{-n^{{\beta}'}}).
\end{aligned}
\end{equation}
Steps (a)-(f) are justified as follows.\\
(a) follows from observing that ${Q}(x^n)={P}(x^n)$ and using the equality \\
\begin{equation} \label{eq:Abel}
\prod_{i=1}^n B_i-\prod_{i=1}^n A_i =\sum_{i=1}^n (B_i-A_i)\Big(\prod_{j=1}^{i-1}A_j \Big)\Big(\prod_{j=i+1}^{n}B_j \Big)
\end{equation}
(\cite{kor09a}, Lemma 3.5);\\
(b) The triangle inequality and ${Q}(u_i|u^{i-1},x^n)={P}(u_i|u^{i-1},x^n)$ for $i \in \mathcal{I}_T$; \\
(c) Definition of $u_i(u^{i-1})$ in \eqref{eq:func1}; \\
(d) The Cauchy-Schwarz inequality; \\
(e) By definition of  $u_i(u^{i-1})$ in \eqref{eq:func1} we have ${P}(u_i(u^{i-1}) \oplus 1,u^{i-1}) \leq {P}(u_i(u^{i-1}) ,u^{i-1}),$
and 
   $$
   {P}(u_i(u^{i-1}) \oplus 1,u^{i-1}) {P}(u_i(u^{i-1}) ,u^{i-1})={P}(0,u^{i-1}) {P}(1,u^{i-1});
   $$
(f) Since ${P}(0|U^{i-1},X^n)+{P}(1|U^{i-1},X^n)=1$, we have  $\sqrt{{P}(0|U^{i-1},X^n){P}(1|U^{i-1},X^n)} \leq 1/2$. 

\vspace*{.1in}
Similarly, the $L_1$ distance between ${P}_{W^nT^nX^n}$ and ${Q}_{W^nT^nX^n}$ can be bounded as
\begin{align*}
\|{P}_{W^nT^nX^n}-{Q}_{W^nT^nX^n}\|_1 &=\|{P}_{V^nT^nX^n}-{Q}_{V^nT^nX^n}\|_1 \\ 
&=\sum_{v^n,t^n,x^n}|{P}(v^n|t^n,x^n){P}(t^n,x^n)-{Q}(v^n|t^n,x^n){P}(t^n,x^n)\\
&\hspace*{.5in}+{Q}(v^n|t^n,x^n){P}(t^n,x^n)-{Q}(v^n|t^n,x^n){Q}(t^n,x^n)| \\
&\leq \sum_{v^n,t^n,x^n}|({P}(v^n|t^n,x^n)-{Q}(v^n|t^n,x^n)){P}(t^n,x^n)| \\
&\hspace*{.5in}+\sum_{v^n,t^n,x^n}\Big({Q}(v^n|t^n,x^n)|{Q}(t^n,x^n)-{P}(t^n,x^n)| \Big) \\
&\overset{(g)}{\leq}O(2^{-n^{{\beta}'}})+||{P}_{T^nX^n}-{Q}_{T^nX^n}||_1 \\
&\overset{(h)}{=}O(2^{-n^{{\beta}'}})+||{P}_{U^nX^n}-{Q}_{U^nX^n}||_1 =O(2^{-n^{{\beta}'}})
\end{align*}
where (g) is obtained in the same way as \eqref{proof1} and (h) follows from the fact that the mapping between $U^n$ and $W^n$ is bijective.
Therefore, we obtain
\begin{align*}
{E}[D_{1,n} (U_{\mathcal{H}_{T|X}})]&=E_{Q}[d(X^n,T^n)] \\
&\leq E_P[d(X^n,T^n)] + (\max d(t,x))\|{P}_{X^nT^n}-{Q}_{X^nT^n}\|_1 \\
&\leq D_1+O(2^{-n^{{\beta}'}})\\
{E}[D_{2,n} (U_{\mathcal{H}_{T|X}},V_{\mathcal{H}_{W|TX}})]&={E}_{{Q}}[d(W^n,X^n)] \\
&\leq {E}_{{P}}[d(X^n,W^n)] + (\max d(w,x))\|{P}_{W^nT^nX^n}-{Q}_{W^nT^nX^n}\|_1\\
& \leq D_2+O(2^{-n^{{\beta}'}}).
\end{align*}
This concludes the proof.
\end{proof}
{\em Remark 1:} { 
We can extend our result to a slightly more general case.  Namely, suppose that the source is not be successively refinable and in particular, the Markov condition \eqref{eq:Markov} is not satisfied. In this case the achievable rate pairs are characterized by the following result of Rimoldi \cite{Rimoldi94}  (see also Koshelev  \cite{Koshelev80}). The rate pair $(R_1,R_2)$ is achievable with distortions $D_1,D_2$ if and only if there exists a conditional distribution 
$P_{TW|X}$ such that the following four inequalities are satisfied
\begin{equation}\label{gencond}
R_1 \geq I(X;T), \quad E_{XT}[d(X,T)] \leq D_1, \quad R_2 \geq I(X;WT), \quad E_{XW}[d(X,W)] \leq D_2.
\end{equation}
We note that the scheme described in this section implies that the polar code consruction achieves these
rate values for given distortion levels $D_1,D_2.$
}

\vspace*{.1in} {\em Remark 2:}
{The concept of successive refinement can be extended to a multilevel representation of the source $X$ in a natural way; see \cite{Koshelev80,Koshelev81}. Roughly speaking, given distortions $D_1 \geq D_2 \geq ... \geq D_t$, if  
the rate region $(R(D_1),R(D_2),...,R(D_t))$ is achievable, then the source is said to be successively refinable at $t\ge 2$ levels.
It is easy to see that the source is successively refinable at $t$ levels if and only if 
it is successively refinable between any two consecutive levels. As a result, the region of achievable rates for the $t$-step refinement of $X$ can be achieved by consecutively using the coding scheme presented above.

\vspace*{.1in} {\em Remark 3:}
As noted in \cite{Equitz91}, the successive refinement problem is a particular case of the problem of multiple descriptions of the source. The region of achievable rates of the multiple description problem was established in
\cite{Ahlswede85,ElGamal91}. Our considerations can be easily extended to this case, giving an explicit construction of codes
attaining the rate region of multiple descriptions. We note that a recent preprint \cite{Sahebi14} also discusses an approach to achieving this rate region using polar codes relying on partitions of the form \eqref{eq:part}.

\section{Successive Refinement for the Wyner-Ziv problem}\label{sect:WZ}
The Wyner-Ziv version of the distributed source coding problem \cite{Wyner76} assumes that the decoder
is provided with the side information in the form of a random variable $Z$ that is correlated
with the source $X$. The correlation is expressed through a joint distribution $P_{XZ}$ known to both 
the encoder and the decoder. And the encoder is not aware of the realization of $Z$. As shown in \cite{Wyner76}, the source sequence $X^n$ can be reproduced with
distortion $D$ if the number of messages used to represent the source is of order $\exp(n R_{X|Z}(D)),$ where
   \begin{equation} \label{eq:WZ}
R_{X|Z}(D)=\min I(X;T|Z)
     \end{equation}
and the minimization is over all random variables $T$ such that 
$T \rightarrow X \rightarrow Z$ is a Markov chain and such that there exists a 
function $f$ acting on  $(T, Z)$ that satisfies 
\begin{equation}
E_{XTZ}d(X,f(T,Z)) \leq D.
\end{equation}

Korada \cite{kor09a} suggested a polar coding scheme that attains the rate \eqref{eq:WZ} under the
assumption that the side information has the form $Z=X+\theta,$ where $\theta$ is a Bernoulli random variable.
In the first part of this section we observe that the ideas developed above enable one to design a
constructive scheme for the general version of the Wyner-Ziv problem.
\subsection{Polar codes for distributed source coding}
Let $(T,X,Z)$ be random variables that achieve the minimum in \eqref{eq:WZ} . As usual, let $(T^n,X^n,Z^n)$ denote $n$ independent copies of $(T,X,Z)$. 
Define $U^n=T^nG_n$ and denote the joint distribution of $(U^n,T^n,X^n,Z^n)$ by 
${P}_{U^nT^nX^nZ^n}$. 
Define the index subsets $\mathcal{H}_{T|Z}$, $\mathcal{L}_{T|Z}$, $\mathcal{H}_{T|XZ}$, $\mathcal{L}_{T|XZ}$ 
in the same way as in \eqref{def2}. 
Let us partition the set of indices as follows: $[n]=\mathcal{H}_{T|XZ}\cup \mathcal{L}_{T|Z}\cup \mathcal{I}_T,$ where $\mathcal{I}_T\triangleq(\mathcal{L}_{T|Z} \cup \mathcal{H}_{T|XZ})^c$ is the subset of
indices that carry the information. Observe that
\begin{equation}\label{eq:rate}
\lim_{n \to \infty} \frac{1}{n}|\mathcal{I}_T|=R_{X|Z}(D).
\end{equation}
We proceed analogously to \eqref{eq:random1}-\eqref{eq:func1}, constructing the sequence $u_i$ that communicates the information. For $i\in \mathcal{H}_{T|XZ}$ assign the bit values $u_i=0$ or $1$ with probability $1/2$ independently of $x^n, z^n$ and each other and make them available both to the encoder and decoder.
The remaining bits $u_i,i\in \mathcal{H}_{T|XZ}^c$ are assigned successively as follows. If $i\in \cI_T,$ choose
$u_i$ in a randomized way according to the distribution
         \begin{equation} \label{random3}
        \Pr(u_i=a)=P_{U_i|U^{i-1}Z}(a|u^{i-1},x_n), \quad a=0,1
        \end{equation}
and if $i \in \mathcal{L}_{T|Z},$ put
\begin{equation} \label{func3}
u_i=u_i(u^{i-1},x^n)\triangleq \underset{a \in \{0,1\}}{\arg\max \text{ }} {P}_{U_i|U^{i-1}X^n}(a|u^{i-1},x^n).
\end{equation}
The decoder is provided with the sequence $u_{\mathcal{I}_T}$ and constructs
an estimate of the bits $u_{\mathcal{L}_{T|Z}}$ successively by setting 
\begin{equation} \label{func4}
\hat u_i=\hat u_i(u^{i-1},z^n)\triangleq\underset{a \in \{0,1\}} {\mathrm{argmax} \text{ }} {P}_{U_i|U^{i-1}Z^n}(a|u^{i-1},z^n).
\end{equation}
Then the decoder calculates $t^n=\hat u^nG_n$ and outputs the reproduction sequence $s^n=(f(t_1,z_1),f(t_2,z_2),...,f(t_n,z_n))$. By assumptions, the communication rate 
approaches $R_{X|Z}(D)$ \eqref{eq:rate}. 
For a certain choice of the frozen bits $u_{\mathcal{H}_{T|XZ}}$, the average distortion is given by
\begin{equation*}
D_n (u_{\mathcal{H}_{T|XZ}})
=E_{X^nZ^n} \Big[ {E} [d(X^n,s^n(u_{\mathcal{H}_{T|XZ}},X^n,Z^n))] \Big]
\end{equation*}
where the inner expectation is computed over the randomized choice of the bits in $\cI_T$ via \eqref{random3}.

Next we show that for some choice of the frozen bits $u_{\cH_{T|XZ}}$ this construction attains the desired distortion level. For this we must show that the sequences $u^n$ and $\hat u^n$ coincide with high probability. 
This is not obvious because the encoder has no access to the side information $Z^n$ 
while the decoder has no access to the source $X^n.$
\begin{theorem} 
For any $0 < {\beta}' < \beta < 1/2$ 
   $$
   E[D_n (U_{\mathcal{H}_{T|XZ}})] \leq D+O(2^{-n^{{\beta}'}})
   $$
   where the expectation is computed over the choice of the frozen bits $u_{\mathcal{H}_{T|XZ}}.$
    Consequently, there exists a choice such that $D_n (u_{\mathcal{H}_{T|XZ}}) \leq D+O(2^{-n^{{\beta}'}})$.
\end{theorem}
\begin{proof}
Define a function $U^i:{\{0,1\}}^{|\mathcal{L}_{T|Z}^c \cap [i]|} \times {\mathcal{X}}^n \to {\{0,1\}}^i$ as follows: it computes the values $u_j,j\in \mathcal{L}_{T|Z}^c \cap [i]$ successively
using the rule \eqref{func3}. 
Define the function $\hat{U}^i:{\{0,1\}}^{|\mathcal{L}_{T|Z}^c \cap [i]|} \times {\mathcal{Z}}^n \to {\{0,1\}}^i$ in a similar way: it asignes the values $\hat u_j, j\in \mathcal{L}_{T|Z} \cap [i]$ successively using \eqref{func4}.

The proof relies on establishing the proximity of distributions of the sequences $u^n$ and $\hat u^n$ available
to the encoder and the decoder, respectively. 
Define the distribution ${\hat Q}_{U^nX^nZ^n}$ as follows:
  $$ 
{\hat Q}_{U^nX^nZ^n}(u^n,x^n,z^n)={P}_{X^nZ^n}(x^n,z^n) \Big(\prod_{i=1}^n {\hat Q}_{U_i|U^{i-1}X^nZ^n}(u_i|u^{i-1},x^n,z^n) \Big)
   $$
   where
   \begin{align*}
{\hat Q}_{U_i|U^{i-1}X^nZ^n}(u_i|u^{i-1},x^n,z^n) 
= \begin{cases} \frac{1}{2} & \text{if } i \in \mathcal{H}_{T|XZ} \\ \mathbbm{1}[{P}_{U_i|U^{i-1}Z^n}(u_i|u^{i-1},z^n) > {P}_{U_i|U^{i-1}Z^n}(u_i \oplus 1|u^{i-1},z^n)] & \text{if } i \in \mathcal{L}_{T|Z} \\ {P}_{U_i|U^{i-1} X^n}(u_i|U^{i-1}(u_{\mathcal{L}_{T|Z}^c \cap [i-1]},x^n),x^n) & \text{if } i \in \mathcal{I}_T. \end{cases}
\end{align*}
The theorem will follow if we show that
\begin{equation} \label{proof2}
||\hat Q_{U^nX^nZ^n}-{P}_{U^nX^nZ^n}||_1 \leq O(2^{-n^{{\beta}'}}).
\end{equation}
Define another joint distribution ${Q}_{U^nX^nZ^n}$ as follows
    \begin{equation}\label{eq:aj}
{Q}_{U^nX^nZ^n}(u^n,x^n,z^n)={P}_{X^nZ^n}(x^n,z^n) \Big(\prod_{i=1}^n {Q}_{U_i|U^{i-1}X^nZ^n}(u_i|u^{i-1},x^n,z^n) \Big)
   \end{equation}
where
   \begin{align*}
{Q}_{U_i|U^{i-1}X^nZ^n}(u_i|u^{i-1},x^n,z^n) 
= \begin{cases} \frac{1}{2} & \text{if } i \in \mathcal{H}_{T|XZ} \\ \mathbbm{1}[{P}_{U_i|U^{i-1}Z^n}(u_i|u^{i-1},z^n) > {P}_{U_i|U^{i-1}Z^n}(u_i \oplus 1|u^{i-1},z^n)] & \text{if } i \in \mathcal{L}_{T|Z} \\ {P}_{U_i|U^{i-1} X^n}(u_i|u^{i-1},x^n) & \text{if } i \in \mathcal{I}_T. \end{cases}
\end{align*}
It can be easily verified that we can also use
\begin{equation} \label{alterdef}
{Q}_{U_i|U^{i-1}X^nZ^n}(u_i|u^{i-1},x^n,z^n)={P}_{U_i|U^{i-1} X^n}(u_i|\hat{U}^{i-1}(u_{\mathcal{L}_{T|Z}^c \cap [i-1]},z^n),x^n) \text{ if }  i \in \mathcal{I}_T 
\end{equation}
and get the same distribution ${Q}_{U^nX^nZ^n}$ as in \eqref{eq:aj}. The $L_1$ distance between ${Q}_{U^nX^nZ^n}$ and ${P}_{U^nX^nZ^n}$ can be bounded in exactly the same way as in \eqref{proof1}, and we obtain
\begin{equation}
||{Q}_{U^nX^nZ^n}-{P}_{U^nX^nZ^n}||_1 \leq O(2^{-n^{{\beta}'}}).
\end{equation}

Let $\mathcal{L}_{T|Z}=\{k_1,...,k_{|\mathcal{L}_{T|Z}|}\}$, where $k_1<...<k_{|\mathcal{L}_{T|Z}|}$. Define the sets ${\{\mathcal{A}_{k_i}\}}_{i=1}^{|\mathcal{L}_{T|Z}|}$ and ${\{\tilde{\mathcal{A}}_{k_i}\}}_{i=1}^{|\mathcal{L}_{T|Z}|}$ as follows:
    \begin{align*}
\mathcal{A}_{k_i} =\{(u^{k_i-1},x^n,z^n)| & \hat{U}^j(u_{\mathcal{L}_{T|Z}^c \cap [j-1]},z^n)=U^j(u_{\mathcal{L}_{T|Z}^c \cap [j-1]},x^n) \text{ for all } j<k_i \\
&\text{and }  \hat{U}^{k_i}(u_{\mathcal{L}_{T|Z}^c \cap [k_i-1]},z^n) \neq U^{k_i}(u_{\mathcal{L}_{T|Z}^c \cap [k_i-1]},x^n) \}
    \end{align*}
    $$
\tilde{\mathcal{A}}_{k_i}=\{(u^n,x^n,z^n)|(u^{k_i-1},x^n,z^n) \in \mathcal{A}_{k_i} \}.
    $$
By definition  the sets ${\{\tilde{\mathcal{A}}_{k_i}\}}_{i=1}^{|\mathcal{L}_{T|Z}|}$ for different $i$ are pairwise disjoint. 
If $(u^n,x^n,z^n) \in \tilde{\mathcal{A}}_{k_j}$, then 
  $$
  {Q}_{U_i|U^{i-1}X^nZ^n}(u_i|u^{i-1},x^n,z^n)=\hat Q_{U_i|U^{i-1}X^nZ^n}(u_i|u^{i-1},x^n,z^n)
  $$
   for all $i<k_j,$ where ${Q}_{U_i|U^{i-1}X^nZ^n}$ is given by \eqref{alterdef}. 
   Further, if $(u^n,x^n,z^n) \in ( \cup_{i=1}^{|\mathcal{L}_{T|Z}|}\tilde{\mathcal{A}}_{k_i})^c$, then ${Q}_{U^nX^nZ^n}=\hat Q_{U^nX^nZ^n}$. This enables us to bound the $L_1$ distance between the distributions $Q_{U^nX^nZ^n}$ and $\hat Q_{U^nX^nZ^n}$ as follows:
\begin{equation}
\begin{aligned}
\|Q_{U^nX^nZ^n}-&\hat Q_{U^nX^nZ^n}\|_1 \\
&= \sum_{i=1}^{|\mathcal{L}_{T|Z}|} \Big( \sum_{(u^n,x^n,z^n) \in \tilde{\mathcal{A}}_{k_i}}|{Q}(u^n,x^n,z^n)-\hat Q(u^n,x^n,z^n)| \Big) \\
&= \sum_{i=1}^{|\mathcal{L}_{T|Z}|} \biggl( \sum_{(u^n,x^n,z^n) \in \tilde{\mathcal{A}}_{k_i}} \Big( {Q}(u^{k_i-1},x^n,z^n) \Big| {Q}(u_{k_i}^n|u^{k_i-1},x^n,z^n)-\hat Q(u_{k_i}^n|u^{k_i-1},x^n,z^n) \Big| \Big) \biggr) \\
&=\sum_{i=1}^{|\mathcal{L}_{T|Z}|} \biggl( \sum_{(u^{k_i-1},x^n,z^n) \in \mathcal{A}_{k_i}} \Big( {Q}(u^{k_i-1},x^n,z^n) \| {Q}_{U_{k_i}^n|u^{k_i-1},x^n,z^n}-\hat Q_{U_{k_i}^n|u^{k_i-1},x^n,z^n} \|_1 \Big) \biggr) \\
& \overset{(a)}{\leq} 2 \sum_{i=1}^{|\mathcal{L}_{T|Z}|} 
\biggl( \sum_{(u^{k_i-1},x^n,z^n) \in \mathcal{A}_{k_i}} {Q}(u^{k_i-1},x^n,z^n) \biggr)  \\
& \leq 2 \sum_{i=1}^{|\mathcal{L}_{T|Z}|} \Big(  \sum_{(u^{k_i-1},x^n,z^n) \in \mathcal{A}_{k_i}} {P}(u^{k_i-1},x^n,z^n)  + ||{Q}_{U^n,X^n,Z^n}-{P}_{U^n,X^n,Z^n}||_1 \Big) \\
&\overset{(b)}{=} 2 \sum_{i=1}^{|\mathcal{L}_{T|Z}|} \sum_{(u^{k_i-1},x^n,z^n) \in \mathcal{A}_{k_i}}  
\Big( {P}_{U_{k_i}U^{k_i-1}X^nZ^n}(u_{k_i}(u^{k_i-1},x^n) \oplus 1,u^{k_i-1},x^n,z^n)   \\
&  + {P}_{U_{k_i}U^{k_i-1}X^nZ^n}(\hat{u}_{k_i}(u^{k_i-1},z^n) \oplus 1,u^{k_i-1},x^n,z^n) \Big) + 2|\mathcal{L}_{T|Z}| \|{Q}_{U^nX^nZ^n}-{P}_{U^nX^nZ^n}\|_1 \\
& \leq 2 \sum_{i=1}^{|\mathcal{L}_{T|Z}|} \Big( \sum_{u^{k_i-1},x^n}{P}_{U_{k_i}U^{k_i-1}X^n}(u_{k_i}(u^{k_i-1},x^n) \oplus 1,u^{k_i-1},x^n)  \\
&  + \sum_{u^{k_i-1},z^n} {P}_{U_{k_i}U^{k_i-1}Z^n}(\hat{u}_{k_i}(u^{k_i-1},z^n) \oplus 1,u^{k_i-1},z^n) 
\Big) + O(2^{-n^{{\beta}'}}) \\
& \leq \sum_{i=1}^{|\mathcal{L}_{T|Z}|} \Big( Z(U_{k_i}|U^{k_i-1},X^n)+Z(U_{k_i}|U^{k_i-1},Z^n) \Big) + O(2^{-n^{{\beta}'}}) \\
& \overset{(c)}{\leq} O(2^{-n^{{\beta}'}})
\end{aligned}
\end{equation}
where the steps (a)-(c) are justified as follows.\\
(a) The $L_1$ distance between two distributions is always upper bounded by 2; \\
(b) By definition, if $(u^{k_i-1},x^n,z^n) \in \mathcal{A}_{k_i}$ then 
$u_{k_i}\ne \hat u_{k_i},$ and so $\{u_{k_i}\oplus 1,\hat u_{k_i}\oplus 1\}=\{0,1\};$\\
(c) The condition $T \rightarrow X \rightarrow Z$ implies $U^n \rightarrow X^n \rightarrow Z^n$. Thus $Z(U_{k_i}|U^{k_i-1},X^n)=Z(U_{k_i}|U^{k_i-1},X^n,Z^n)$. 

To complete the proof, use the triangle inequality
\begin{align*}
\|\hat Q_{U^nX^nZ^n}-{P}_{U^nX^nZ^n}\|_1 &\leq \|\hat Q_{U^nX^nZ^n}-{Q}_{U^nX^nZ^n}\|_1 +\|{Q}_{U^nX^nZ^n}-{P}_{U^nX^nZ^n}\|_1\\
  &\leq O(2^{-n^{{\beta}'}}).
\end{align*}
\end{proof}

\subsection{Successive refinement for the Wyner-Ziv problem} 
In this section we extend the ideas of the construction of the previous section to the case of distributed
successive refinement. We begin with the following definition.
\begin{definition}\label{def:scsi} {\rm \cite{Steinberg04}} Let $X$ be a discrete memoryless source, $Z$ and $Y$ be the side information available to the decoders at the coarse and the refinement stages, respectively.
Suppose that there exist encoding functions
  \begin{align} \label{eq:sr4}
    &\phi_1:\cX^n\to [M_1]\\
   & \phi_2:\cX^n\to[M_2]\label{eq:sr4-1}
    \end{align}
    and decoding functions
 \begin{align}\label{eq:sr5}
   &\psi_1:[M_1]\times \cZ^n\to \cT^n\\
   &\psi_2:[M_1]\times[M_2] \times \cY^n\to \cT^n\label{eq:sr5-1}
   \end{align} 
such that 
  \begin{align}\label{eq:sr6}
   &E_{X^nZ^n}d(X^n,\psi_1(\phi_1(X^n),Z^n))\le D_1\\
   &E_{X^nY^n}d(X^n,\psi_2(\phi_1(X^n),\phi_2(X^n),Y^n))\le D_2. \label{eq:sr6-1}
  \end{align}
Let  $M_1=2^{nR_1},M_2=2^{n(R_2-R_1)},$ where $(R_1,R_2)$ are the rate values of the encoders.
We say that the rate pair $(R_1,R_2)$ is achievable with distortions $D_1,D_2$ if for any 
$\epsilon_1>0,\epsilon_2>0,\delta>0$ there exists a sufficiently large $n$ such that there exists a coding scheme \eqref{eq:sr4}-\eqref{eq:sr6-1} with block length $n$, rates not exceeding $R_1+\epsilon_1,R_2+\epsilon_1+\epsilon_2$
and distortions $D_1+\delta,D_2+\delta.$

The source $X$ is said to be {\em successively refinable} with distortions $D_1$ and $D_2,$
$D_2\le D_1,$ if the rate pair $(R_{X|Z}(D_1),$ $R_{X|Y}(D_2))$ is achievable.
\end{definition}
As before, the realizations of side informations are only available to decoders, while the joint distribution of $(X,Y,Z)$ is known to both encoder and decoders. For the case where the Markov relation $X \rightarrow Y \rightarrow Z$ holds, Steinberg and Merhav \cite{Steinberg04} gave the following  necessary and sufficient condition 
for successive refinability.
\begin{theorem}\label{thm:SRWZ}
A source $X$ with degraded side information $(Y,Z)$ is successively refinable from $D_1$ to $D_2$ if and only if there exist a pair of random variables $(T,W)$ and a pair of deterministic maps $f_1:\mathcal{T} \times \mathcal{Z} \to \hat{\mathcal{X}}$ and $f_2:\mathcal{W} \times \mathcal{Y} \to \hat{\mathcal{X}}$ such that the following conditions simultaneously hold:\\
1) $R_{X|Z}(D_1)=I(X;T|Z)$ and $E_{XTZ}d(X,f_1(T,Z))\leq D_1$; \\
2) $R_{X|Y}(D_2)=I(X;W|Y)$ and $E_{XWY}d(X,f_2(W,Y))\leq D_2$; \\
3) $(T,W) \rightarrow X \rightarrow Y \rightarrow Z$ form a Markov chain; \\
4) $T \rightarrow (W,Y) \rightarrow X$ form a Markov chain; \\
5) $I(T;Y|Z)=0$.
\end{theorem}

Using the results in \cite{Steinberg04} and \eqref{eq:WZ}, we note that Conditions 1)-5) above imply 
\begin{gather} \label{eq:cond1}
H(T|Z)-H(T|X,Z)=I(X;T|Z)=R_{X|Z}(D_1)\\
H(W|T,Y)-H(W|X,T,Y)=I(X;W|T,Y)=R_{X|Y}(D_2)-R_{X|Z}(D_1).\label{eq:cond2}
\end{gather}
Also, from condition 3) we have $T \rightarrow Y \rightarrow Z$, while from condition 5) we have $T \rightarrow Z \rightarrow Y$. Then, for any $(y_1,z_1) \in \mathcal{Y} \times \mathcal{Z}$, $(y_2,z_2) \in \mathcal{Y} \times \mathcal{Z},$ and $t \in \mathcal{T}$ we have that 
  $$
  P_{T|YZ}(t|y_1,z_1)={P}_{T|YZ}(t|y_1,z_2)={P}_{T|YZ}(t|y_2,z_2).
  $$
 Thus $T$ is independent of $(Y,Z)$ and \eqref{eq:cond1} reduces to
    $$
I(T;X)=H(T)-H(T|X)=R_{X|Z}(D_1).
     $$
Let $(T,W,X,Y,Z)$ be a quintuple of random variables that satisfy conditions 1)-5) and let $(T^n,W^n,X^n,Y^n,Z^n)$ be its $n$ 
independent copies. 
Further, let $U^n=T^nG_n$ and $V^n=W^nG_n$. We denote the joint distribution of the $n$-sequences by 
${P}_{U^nT^nV^nW^nX^nY^nZ^n}$  and   use various conditional and marginal distributions derived from it. Define the coordinate subsets $\mathcal{H}_T$, $\mathcal{L}_T$, $\mathcal{H}_{T|X}$, $\mathcal{L}_{T|X}$, $
\mathcal{H}_{W|TY}$, $\mathcal{L}_{W|TY}$, $\mathcal{H}_{W|XTY}$, $\mathcal{L}_{W|XTY}$ in the same way as \eqref{def1} and \eqref{def2}. 
Further, define the information sets as $\mathcal{I}_T=(\mathcal{L}_T \cup \mathcal{H}_{T|X})^c$ and $\mathcal{I}_W=(\mathcal{L}_{W|TY} \cup \mathcal{H}_{W|XTY})^c$
and observe that
\begin{gather*}
\lim_{n \to \infty} \frac{1}{n}|\mathcal{I}_T|=I(X;T)=R_{X|Z}(D_1) \\
\lim_{n \to \infty} \frac{1}{n}|\mathcal{I}_W|=I(X;W|T,Y)=R_{X|Y}(D_2)-R_{X|Z}(D_1).
\end{gather*}
The coding scheme is similar to the previous section. The bits $u_i$, $i \in \mathcal{H}_{T|X}$, are set to $0$ or $1$ with probability 1/2 independently of each other and of $(X^n,Y^n,Z^n),$ and are known to both the encoder and the decoder. 
The remaining bits of the sequence $u^n$ are determined successively as follows. 
For $i \in \mathcal{I}_T$ we find
$u_i$ in a randomized way according to the distribution
         \begin{equation} \label{random4}
    \Pr(u_i=a)={P}_{U_i|U^{i-1}X^n}(a|u^{i-1},x^n), \quad a\in\{0,1\}
    \end{equation}
and if $i \in \mathcal{L}_T,$ then we put
   \begin{equation} \label{func5}
u_i= \arg\max_{a\in\{0,1\}}  {P}_{U_i|U^{i-1}}(a|u^{i-1}).
   \end{equation}
After determining the entire sequence $u^n$, the encoder calculates the sequence $t^n=u^nG_n$ and uses it to determine $v^n$ according to the following rule. 
The bits $v_i$, $i \in \mathcal{H}_{W|XTY}$, are drawn uniformly from $\{0,1\}$ independently of each other, and of $u_{\mathcal{H}_{T|X}}$ and $(X^n,Y^n,Z^n)$. Also make these frozen bits known to both the encoder and the decoder.
If $i \in \mathcal{I}_W,$ then $v_i$ is chosen in a randomized way from the distribution
    \begin{equation} \label{random5}
    \Pr(v_i=a)={P}_{V_i|V^{i-1}T^nX^n}(a|v^{i-1},t^n,x^n) , \quad a\in\{0,1\}
         \end{equation}
and if $i \in \mathcal{L}_{W|TY}$ then
  $$
v_i=\arg\max_{a\in\{0,1\}} {P}_{V_i|V^{i-1}T^nX^n}(a|v^{i-1},t^n,x^n).
   $$
This concludes the description of the encoding scheme.

The decoder constructs reproduction sequences of the source upon being provided with
the sequences $u_{\mathcal{I}_T}$ and $v_{\mathcal{I}_W}.$ We will assume that the functions $f_1$ and $f_2$ 
from Theorem \ref{thm:SRWZ} are available to the decoder.
At the coarse layer the decoder determines the sequence $u_{\mathcal{L}_T}$ 
successively using the rule \eqref{func5}. 
Then the decoder calculates the sequence $ t^n= u^nG_n$ and forms the reproduction sequence 
$r^n=(f_1( t_1,z_1),...,f_1( t_n,z_n))$. 
At the refinement layer, the decoder uses $ t^n$ to determine the sequence 
$v_{\mathcal{L}_{W|TY}}$ successively according to the rule
  $$ 
v_i=\arg\max_{a\in\{0,1\}} {P}_{V_i|V^{i-1}T^nY^n}(a| v^{i-1}, t^n,y^n).
  $$ 
Upon finding $v^n,$ the decoder computes $w^n=v^nG_n$. The reproduction sequence at refinement stage is found as $s^n=(f_2( w_1,y_1),...,f_2( w_n,y_n))$. 

By assumptions the communication rates approach $R_{X|Z}(D_1)$ and $R_{X|Y}(D_2)$ \eqref{eq:cond1}, \eqref{eq:cond2}. For a fixed assignment of the frozen bits $u_{\mathcal{H}_{T|X}}$ and $v_{\mathcal{H}_{W|XTY}}$, the average distortions are given by
\begin{gather*}
D_{1,n} (u_{\mathcal{H}_{T|X}})=E_{X^nZ^n} \Big[ {E} [d(X^n,r^n(u_{\mathcal{H}_{T|X}},X^n,Z^n))] \Big]\\
D_{2,n} (u_{\mathcal{H}_{T|X}},v_{\mathcal{H}_{W|XTY}})= E_{X^nY^n} \Big[ {E} [d(X^n,s^n(u_{\mathcal{H}_{T|X}},v_{\mathcal{H}_{W|XTY}},X^n,Y^n))] \Big]
\end{gather*}
where the inner expectation is taken with respect to the distributions \eqref{random4} and \eqref{random5}, respectively. 

Next we show that expected distortions computed by averaging over all choices of the frozen bits $u_{\mathcal{H}_{T|X}}$ and $v_{\mathcal{H}_{W|XTY}}$
\remove{We consider the average of $D_{1,n} (u_{\mathcal{H}_{T|X}})$ and $D_{2,n} (u_{\mathcal{H}_{T|X}},v_{\mathcal{H}_{W|XTY}})$ over all choices for frozen bits in the following theorem.}
are close to the chosen levels $D_1$ and $D_2.$
\begin{theorem} For any $0 < {\beta}' < \beta < 1/2$ 
     \begin{align*}
     {E}[D_{1,n} (U_{\mathcal{H}_{T|X}})] \leq D_1+O(2^{-n^{{\beta}'}})\\
       {E}[D_{2,n} (U_{\mathcal{H}_{T|X}},V_{\mathcal{H}_{W|XTY}})] \leq D_2+O(2^{-n^{{\beta}'}}).
     \end{align*}
Consequently, there exists a choice of frozen bits $u_{\mathcal{H}_{T|X}}$ and $v_{\mathcal{H}_{W|XTY}}$ 
such that the distortion values of the reproduction sequences approach the values $D_1$ and $D_2.$
\end{theorem}
\begin{proof}
The decoded sequences $(U^n, V^n,T^n, W^n)$ are random functions of $(X^n,Y^n,Z^n)$. 
Denote their joint distribution by $\hat Q_{U^n  V^n T^n  W^n X^n Y^n Z^n}$. To prove the theorem we only need to bound the $L_1$ distance between $\hat Q_{W^n T^n X^n Y^n Z^n}$ and ${P}_{W^nT^nX^nY^nZ^n}$.
 $$
 \hat Q_{U^nX^nY^nZ^n}(u^n,x^n,y^n,z^n)={P}_{X^nY^nZ^n}(x^n,y^n,z^n) \prod_{i=1}^n \hat Q_{U_i|U^{i-1}X^nY^nZ^n}(u_i|u^{i-1},x^n,y^n,z^n) 
 $$
where
  \begin{align*}
\hat Q_{U_i|U^{i-1}X^nY^nZ^n}(u_i|u^{i-1},x^n,y^n,z^n)&= 
\hat Q_{U_i|U^{i-1}X^n}(u_i|u^{i-1},x^n) \\
& = \begin{cases} \frac{1}{2} & \text{if } i \in \mathcal{H}_{T|X} \\ \mathbbm{1}[{P}_{U_i|U^{i-1}}(u_i|u^{i-1}) > {P}_{U_i|U^{i-1}}(u_i \oplus 1|u^{i-1})] & \text{if } i \in \mathcal{L}_T \\ {P}_{U_i|U^{i-1} X^n}(u_i|u^{i-1},x^n) & \text{if } i \in \mathcal{I}_T. \end{cases}
\end{align*}

The $L_1$ distance between the distributions $\hat Q_{T^nX^nY^nZ^n}$ and ${P}_{T^nX^nY^nZ^n}$ can be bounded as follows
\begin{align*}
 \|\hat Q_{T^nX^nY^nZ^n}&-{P}_{T^nX^nY^nZ^n}\|_1 = \|\hat Q_{U^nX^nY^nZ^n}-{P}_{U^nX^nY^nZ^n}\|_1 \\
&\overset{\eqref{eq:Abel}}{=} \sum_{u^n,x^n,y^n,z^n} \Big| \sum_i(\hat Q(u_i|u^{i-1},x^n,y^n,z^n)-P(u_i|u^{i-1},x^n,y^n,z^n) {P}(x^n,y^n,z^n) \\
& \hspace*{.5in}\times\Big( \prod_{j=1}^{i-1}{P}(u_j|u^{j-1},x^n,y^n,z^n) \Big) \Big( \prod_{j=i+1}^n \hat Q(u_j|u^{j-1},x^n,y^n,z^n) \Big) \Big| \\
&\overset{(a)}{\leq} \sum_{i \in \mathcal{I}_T^c} \sum_{u^i,x^n} \Big( \Big| \hat Q(u_i|u^{i-1},x^n)-P(u_i|u^{i-1},x^n) \Big| {P}(u^{i-1},x^n) \Big) \\
& {\leq} O(2^{-n^{{\beta}'}}).
\end{align*}
where
(a) holds true because $T \rightarrow X \rightarrow Y \rightarrow Z$ implies 
$U^n \rightarrow X^n \rightarrow Y^n \rightarrow Z^n$. 
Thus ${P}(u_i|u^{i-1},x^n,y^n,z^n)={P}(u_i|u^{i-1},x^n)$.
The last equality is obtained in the same way as the analogous result in \eqref{proof1}.

The $L_1$ distance between the distributions $\hat Q_{W^nT^nX^nY^nZ^n}$ and $P_{W^nT^nX^nY^nZ^n}$ is
 bounded as follows
\begin{align*}
\|\hat Q_{W^nT^nX^nY^nZ^n}&-{P}_{W^nT^nX^nY^nZ^n}\|_1 
= \|\hat Q_{V^nT^nX^nY^nZ^n}-P_{V^nT^nX^nY^nZ^n}\|_1 \\
& \leq \sum_{v^n,t^n,x^n,y^n,z^n} \Big( \hat Q(v^n|t^n,x^n,y^n,z^n) \Big| \hat Q(t^n,x^n,y^n,z^n)- P(t^n,x^n,y^n,z^n) \Big| \Big)  \\
&  \hspace*{.5in}+ \sum_{v^n,t^n,x^n,y^n,z^n} \Big( \Big|\hat Q(v^n|t^n,x^n,y^n,z^n) - P(v^n|t^n,x^n,y^n,z^n) \Big| {P}(t^n,x^n,y^n,z^n) \Big)  \\
&= \|\hat Q_{T^nX^nY^nZ^n}-P_{T^nX^nY^nZ^n}\|_1 \\
& \hspace*{.5in}+ \sum_{v^n,t^n,x^n,y^n,z^n} \Big( \Big|\hat Q(v^n|t^n,x^n,y^n,z^n) - P(v^n|t^n,x^n,y^n,z^n) \Big| {P}(t^n,x^n,y^n,z^n) \Big) \\
& {\leq} O(2^{-n^{{\beta}'}})
\end{align*}
where the last inequality follows the same steps as the proof of \eqref{proof2}.
\end{proof}
This shows that the polar coding scheme described above supports communication for successive refinement 
with side information as given in Definition~\ref{def:scsi}.

In conclusion we note that a concurrent work \cite{Gulcu14} relies on the same general setting as this paper, more specifically, on source coding with side information, to propose and analyze a construction of polar codes for distributed computation of functions in certain two- and multi-terminal networks.

\bibliographystyle{IEEEtran}
\bibliography{successive}
\end{document}